\documentclass[11pt]{article} % use larger type; default would be 10pt

\usepackage[utf8]{inputenc} % set input encoding (not needed with XeLaTeX)

\usepackage{graphicx} % support the \includegraphics command and options

\usepackage{booktabs} % for much better looking tables
\usepackage{array} % for better arrays (eg matrices) in maths
\usepackage{paralist} % very flexible & customisable lists (eg. enumerate/itemize, etc.)
\usepackage{verbatim} % adds environment for commenting out blocks of text & for better verbatim
\usepackage{subfig} % make it possible to include more than one captioned figure/table in a single float
\usepackage{tikz-cd}
\usepackage{circuitikz}
\usepackage{ textcomp }
\usepackage{macros}

\usepackage{fancyhdr} % This should be set AFTER setting up the page geometry
\usepackage[affil-it]{authblk}
\usepackage[nottoc,notlof,notlot]{tocbibind} % Put the bibliography in the ToC
\usepackage[titles,subfigure]{tocloft} % Alter the style of the Table of Contents

\usepackage{authblk}
\usepackage{amsthm}
\theoremstyle{plain}
\newtheorem{theorem}{Theorem}
\newtheorem{claim}{Claim}

\theoremstyle{definition}
\newtheorem{definition}{Definition}
\def\rem#1{{\marginpar{\raggedright\scriptsize #1}}}

\title{Optimizing relinearization in circuits for homomorphic encryption}

\author[]{Hao Chen}
\affil[]{Microsoft Research \\
\texttt{haoche@microsoft.com}}

\date{}
\begin{document}
\maketitle  % Activate to display a given date or no date (if empty),
         % otherwise the current date is printed 

\begin{abstract}
Fully homomorphic encryption (FHE) allows an untrusted party to evaluate arithmetic circuits,  i.e., perform additions and multiplications on encrypted data, without having the decryption key.

One of the most efficient class of FHE schemes include BGV/FV schemes, which are based on the hardness of the RLWE problem. They share some common features: ciphertext sizes grow after each homomorphic multiplication; multiplication is much more costly than addition, and the cost of homomorphic multiplication scales linearly with the input ciphertext sizes. Furthermore, there is a special {\it relinearization} operation that reduce the size of a ciphertext, and the cost of relinearization is on the same order of magnitude as homomorpic multiplication. This motivates us to define a discrete optimization problem, which is to decide where (and how much) in a given circuit to relinearize,  in order to minimize the total computational cost. 

In this paper, we formally define the {\it relinearize problem}. We prove that the problem is NP-hard. In addition, in the special case where each vertex has at most one outgoing edge, we give a polynomial-time algorithm. 
\end{abstract}

\section{Introduction}

Fully homomorphic encryption (FHE) is an encryption technique which allows any untrusted party to evaluate functions on encrypted data without the decryption key.  As a typical application, FHE allows a client to outsource computation to an untrusted cloud. It has generated interest in fields such as health and finance, due to the need to analyze sensitive data without having access to the data itself. Since Gentry introduced the first FHE scheme in 2009, there has been a line of work that proposed new FHE schemes with improved efficiency, among which two of the most widely used schemes are \cite{brakerski2014leveled} and its scale-invariant counterpart \cite{fan2012somewhat}. Implementations of these schemes include \cite{halevi2014algorithms}, \cite{chensimple}, and \cite{aguilar2016nfllib}. There has been numerous work that design applications based on these schemes. Some of them (\cite{gilad2016cryptonets,bos2017privacy}) evaluate machine learning models on encrypted data. Others use FHE to design secure protocols such as private information retrieval \cite{melchor2016xpir} and private set intersection \cite{chen2017fast}.

Unfortunately, in these schemes homomorphic operations are still several-orders of magnitude slower than performing the same operation on plaintexts. Therefore, any optimization in the computation time has great interest. In order to use FHE to evaluate a function, one first needs to express the function as an arithmetic circuit. The circuit is represented as a direct acyclic graph with each vertex being either an input, an output, or an arithmetic operation such as multiplication and addition. In both schemes mentioned above, a fresh ciphertext is a pair of polynomials. When a homomorphic multiplication is performed, the length of the output ciphertext grows. More precisely, if we denote the 
length of a ciphertext $c$ by $l(c)$, then $l(  c_1 \otimes c_2 ) = l(c_1) +  l(c_2) - 1$. The length of the result of a homomorphic addition is the maximum length of the two operands, i.e., $l(  c_1 \oplus c_2 ) = \max\{l(c_1), l(c_2)\}$. 

Rouhgly speaking, the computational cost to perform a homomorphic multiplication scales linearly with its input lengths. In both schemes, 
we can model the amount of work it takes to perform a homomorphic multiplication between two ciphertexts $c_1$ and $c_2$ by 
\[
	k_m (l(c_1) + l(c_2)), 
\] 
where $k_m$ is some scheme-dependent constant. In FHE, there is also a squaring operation, which takes as input an encryption of $x$ and returns an encryption of $x^2$. It has the same cost \footnote{Actually, the cost of squaring $x$ is a constant factor smaller than multiplying $x$ with itself. For simplicity, we will assume that the costs are equal. This simplification does not invalidate the results.} and length effect as multiplication, but only takes one input. 

Homomorphic additions, on the other hand, takes much less time to perform compared to multiplication. Hence in this work we will assume that additions are ``free''.  For the same reason, we will adopt the common notation from the FHE literature, and denote by depth of a circuit by the largest number of multiplication vertices contained in a path. 

Note that it is undesirable to let the ciphertext sizes grow, since it will increase both the computational cost and the storage burden. To control the ciphertext sizes, both schemes support a special operation called \emph{Relinearization}. Effectively, relinearizing a ciphertext means reducing its length, while keeping the underlying message the same. We can use this operation to reduce the length of a ciphertext to any integer between two and its original length.  The cost of relinearization scales linearly with the reduction in ciphertext length. In other words, there exists a constant $k_r$ such that reducing the ciphertext lengths by $i$ takes $i \cdot k_r $ units of work. 

Suppose we are given an arithmetic circuit to perform on encrypted inputs. It is now an optimization problem to decide where and how much to relinearize, in order to minimize the total amount of work, consisting of multiplication cost and relinearization cost. Previous works employ the simple strategy of relinearizing after every multiplication/squaring. In this way, the multiplication costs are kept minimal. However, this strategy is not always optimal, as we will demonstrate in Section~\ref{sec: example}.

\subsection{Roadmap}
In Section 2, we will formally describe the problem and show why this simple strategy can be sub-optimal. In Section 3, we prove that the relinearize problem is NP-hard by reducing from the knapsack problem. Finally, in Section 4, we restrict to the special case where each vertex in the circuit has at most one outgoing edge, and give a polynomial time algorithm. 

\subsection{Related work}

The work \cite{carpov2015armadillo} is an effort to 
find a good circuit representation of a function, in order to minimize the total computation time. 

Bootstrapping is an operation that refreshes the so-called noise in FHE ciphertexts. It is an essential yet expensive operation. The two papers \cite{lepoint2013minimal} and \cite{benhamouda2017optimization} aim at minimizing the total number of bootstrapping operations in a circuit, while keeping the noise from overflowing in order to ensure the final result is correct. In their work, the authors implicitly assume the relinearization is done after every multiplication. Similarly, we will make a simplifying assumption that the boostrapping time is a constant, so that it does not factor into our optimizatoin problem. It will be interesting to combine these works in order to achieve an overall optimization that targets both operations.

{\bf Acknowledgement} The author thanks Rebecca Hoberg and Mohit Singh for helpful discussions in preparing this work. We modify the usual definition of the arithmetic circuits to include the squaring operations in FHE.

\section{Problem Description}

To formally describe our problem, we need to properly define circuits used in FHE applications. 

\begin{definition}
a (squaring-enabled) arithmetic circuit is a directed acyclic graph $G = (V,E)$, where there are three kinds of vertices: input vertices has indegree 0 and outdegree 1; output vertices has indegree $\in \{1, 2\}$ and outdegree 0; add/multiply operation vertices have indegree 2 and outdegree 1; finally, square operation vertices have indegree and outdegree both equal to 1. 
\end{definition}

We will define the relinearize problem as an integer programming problem on  arithmetic circuits. For every vertex $i$, we maintain an integer variable $l^{new}(i)$ (the final length of vertex $i$ during homomorphic evaluation of $G$), and an integer 
variable $x_i$, which indicates the amount of relinearization at $i$. We will denote the two parents of a vertex $i$ by $p_1(i)$ and $p_2(i)$. If $i$ is a squaring vertex, then we set $p_1(i) = p_2(i)$. We denote addition vertices by $\oplus$ and multiplication/square vertices by $\otimes$. To resolve ambiguity, we make the convention that if a $\otimes$ vertex has two distinct parents, then it is understood as a multiplication; otherwise it is a squaring.

Then the {\it relinearize problem} on $G$ is 
\[
	\mbox{minimize } k_r \sum_{i \in V } x_i  + \sum_{i = \otimes} k_m (l^{new}(i) + x_i), 
\]
s.t. 

\begin{align*}
l^{new}(i) \geq 2  & \qquad &\mbox{ for all } i \\ 
l^{new}(i)  = l^{new}(p_1(i)) + l^{new}(p_2(i)) - 1 - x_i  & \qquad & \mbox{if } i = \otimes \\ 
l^{new}(i)  \geq l^{new}(p_1(i))  - x_i  & \qquad & \mbox{if } i = \oplus \\ 
l^{new}(i)  \geq l^{new}(p_2(i))  - x_i  & \qquad & \mbox{if } i = \oplus \\
x_i, l^{new}(i) \in \bZ_{\geq 0}  & \qquad &\mbox{ for all } i \\
\end{align*}

\subsection{An example} \label{sec: example} 
To demonstrate the non-trivality of the relinearize problem,  we consider the following circuit: 
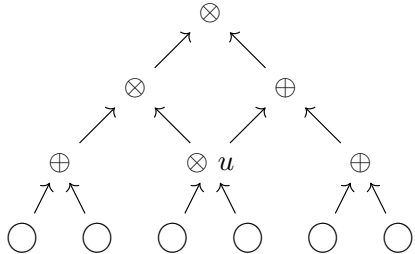
\begin{center}
\begin{tikzpicture}[commutative diagrams/every diagram] 
\node (Pfinal) at (2,2) {$\otimes$}; 
\node (P1) at (1,1) {$\otimes$} ; 
\node (P2) at (3,1) {$\oplus$};
\node (P3) at (0,0) {$\oplus$};
\node (P4) at (2,0) {$\otimes$ $u$};
\node (P5) at (4,0) {$\oplus$};
\node (P6) at (-0.5,-1) {\textbigcircle};
\node (P7) at (0.5,-1) {\textbigcircle};
\node (P8) at (1.5,-1) {\textbigcircle};
\node (P9) at (2.5,-1) {\textbigcircle};
\node (P10) at (3.5,-1) {\textbigcircle};
\node (P11) at (4.5,-1) {\textbigcircle};
\path[commutative diagrams/.cd, every arrow, every label]

(P1) edge node[swap] {} (Pfinal)  
(P2) edge node[swap] {} (Pfinal) 
(P3) edge (P1) 
(P4) edge (P1) 
(P4) edge node {} (P2) 
(P5) edge node {} (P2)
(P6) edge node {} (P3)
(P7) edge node[swap] {} (P3)
(P8) edge node {} (P4)
(P9) edge node[swap] {} (P4)
(P10) edge node  {}  (P5)
(P11) edge node[swap] {} (P5);
\end{tikzpicture}
\end{center}

First, we apply the simple strategy and relinearize at every multiplication vertex. Then the total cost is equal to $12 k_m + 3k_r$. Alternatively, we can choose to only relinearize the vertex $u$. Then the multiplication cost increases to $14k_m$, while the relinearization cost is $k_r$, so the total cost is $14k_m + k_r$. Comparing this with the previous cost, we see that as long as  $k_r > k_m$, the simple strategy is not optimal.

\section{NP-hardness of the Relinearize Problem}

We prove a polynomial reduction from the knapsack problem to the relinearize problem, which establishes that the latter problem is NP-hard.  First we recall the definition of  knapsack problem.

\begin{definition}
Given positive integers $v_1, \ldots, v_n$, $w_1, \ldots ,w_n$ and $W$. The (0-1) knapsack problem is:
\begin{align*}
\mbox{maximize} &\sum_{i=1}^{n} v_i x_i \\
\mbox{subject to } & x_i \in \{0,1\}
\mbox{ and }  \sum w_i x_i \leq W. 
\end{align*}
\end{definition}

For our convenience, we make some modifications to the setting of the relinearize problem. We change the inputs lengths from two to one, and we modify the equation $l(c1 \otimes c2) = l(c1) + l(c2)-1$ to $l(c1 \otimes c2) = l(c1) + l(c2)$.  One can check that under this modificaiton, the length of every vertex is smaller by one. Hence the modified problem is equivalent to the original problem.

%\begin{lemma}
%Under this new formulation every node have $l'(c) = l(c) - 1$. 
%\end{lemma}

%\begin{proof}
%We fix a topological order of the circuit and proceed by induction. The base case is clear since we have the equality holds for the inputs. Now suppose we have a node $i$ with parents $p_1(i), p_2(i)$. Two cases, 
%if $i$ is multiplication, then $l'(i) = l'(p_1(i)) + l'(p_2(i)) + 1 = l(p_1(i)) - 2 + l(p_2(i)) - 2 + 1 = l(i) - 2$.  For addition nodes it is trivial. 
%\end{proof}

To prepare for the main theorem, we make some convenient definitions.  
\begin{definition}
A circuit is of type $\cL(k)$ if it consists of one input vertex, one output vertex, and multiplication/squaring vertices, such that if the first non-input vertex length is reduced from 2 to 1, then the length of the output vertex reduces by $k$.  
\end{definition}

Figure~\ref{fig: L7} is an example of $\cL(7)$. 
\begin{figure}[h!]
\begin{center}
\begin{tikzpicture}[commutative diagrams/every diagram]  
\node (P1) at (2,5) {$\otimes$}; 
\node (P2) at (2,4) {$\otimes$}; 
\node (P3) at (2,3) {$\otimes$}; 
\node (P4) at (2,2) {$\otimes$} ; 
\node (P5) at (2,1) {$\otimes$};
\node (P6) at (2,0) {\textbigcircle};;
\path[commutative diagrams/.cd, every arrow, every label] 
(P2) edge  (P1)  
(P3) edge  (P2) 
(P4) edge  (P3) 
%(P4) edge[bend right]   (P3)
(P4) edge[bend right]   (P1)
(P5) edge[bend left]   (P2)
(P5) edge (P4) 
%(P5) edge[bend right]  (P4)
(P6) edge (P5);
%(P6) edge[bend right]  (P5);
%(P1) edge node {} (P0)  
%(P2) edge node[bend left] {} (P1)
%(P3) edge node[swap] {} (P2);
\end{tikzpicture}
\end{center}
\caption{an example of $\cL(7)$}
\label{fig: L7}
\end{figure}

\begin{lemma} \label{lem: one}
For all integers $k \geq 1$, there exists a circuit of type $\cL(k)$ which has at most $2 \lceil \log k \rceil$ vertices. Moreover, the  cost to evaluate this circuit is bounded above by $4 k_m k \lceil \log(k) \rceil$.
\end{lemma}

\begin{proof}
If $k$ is a power of 2, we can realize $\cL(k)$ by a circuit that does $\log(k) + 1$ consecutive squarings.
The total cost of executing the circuit is $k_m \cdot (2 + 4 + \cdots + 2k) < 4k_m k$.  In general, we can start by  building the circuit $\cL(2^{[\log(k)]})$. Then for every nonzero bit in the binary representation of $k$, we need to add a multiplication vertex.  Since there are at most $\log(k)$ bits, 
we know the number of vertices is at most $2 \log(k)$. 

As for the evaluation cost,  note that each vertex in the circuit has length bounded above by $2k$, hence evaluating it has cost bounded by $2k k_m$. The claim follows because there are at most $2 \lceil \log (k) \rceil$ vertices.
\end{proof}

\iffalse
Here is an example of $\cL(4)$. 
\begin{center}
\begin{tikzpicture}[commutative diagrams/every diagram] 
\node (P0) at (2,5) {$\otimes$}; 
\node (P1) at (2,3) {$\otimes$} ; 
\node (P2) at (2,1) {$\otimes$};
\node (P3) at (2,-1) {I};;
\draw [ ->](2.3,-1) to [out = 30, in = -30] (2.3,1); 
\draw [ ->](1.7,-1) to [out = 150, in = -150] (1.7,1); 
\draw [ ->](2.3,1) to [out = 30, in = -30] (2.3,3); 
\draw [ ->](1.7,1) to [out = 150, in = -150] (1.7,3); 
\draw [ ->](2.3,3) to [out = 30, in = -30] (2.3,5); 
\draw [ ->](1.7,3) to [out = 150, in = -150] (1.7,5); 
%\path[commutative diagrams/.cd, every arrow, every label] 
%(P1) edge node[swap] {} (P0)  
%(P1) edge node {} (P0)  
%(P2) edge node[bend left] {} (P1)
%(P3) edge node[swap] {} (P2);
\end{tikzpicture}
\end{center}
\fi

Next we describe some simple ways to construct new circuits from old ones. 
\begin{definition}
(1) The addition/multiplication of two circuits. Take two circuits $G_1$ and $G_2$ with unique output vertices $v_1$ and $v_2$.  Then $G_1 \boxplus G_2$ (resp. $G_1 \boxtimes G_2$) is the circuit that is the union of $G_1$ and $G_2$, plus an extra addition (resp. multiplication) vertex that has $v_1$ and $v_2$ as parents. 
See Figure~\ref{fig: add} for an example. 

(2) The concatenation of two circuits. Let $G_1, G_2$ be two circuits such that the number of output vertices of $G_1$ is equal to the number of inputs of $G_2$. Then we simply ``feed'' the outputs of $G_1$ to inputs of $G_2$. We denote the resulting circuit by 
$G_1 \curvearrowright G_2$. See Figure~\ref{fig: cat} for an example.

(3) The $K$-repeat of a circuit along a subset of vertices. Let $G$ be a circuit and let $S = \{ s_1, \ldots, s_k \}$ be vertices of $G$. 
Let $K$ be a positive integer. Then we keep the vertices $s_i$ and all their ancestors, and copy the rest of the circuit $K$ times. The resulting circuit is denoted by $G^{(K)}_{S}$. See Figure~\ref{fig: repeat} for an example. 

(4) The gluing of two circuits along a subset of vertices. Let $G_1$ and $G_2$ be two circuits and $S_1, S_2$ be subsets of their vertices, such that the subgraph of $G_1$ consisting of ancestors of $S_1$ (including vertices in $S_1$) is isomorphic to the corresponding subgraph in $G_2$. Then the gluing of $G_1$ and $G_2$ along $S_1, S_2$ is the circuit that contains the common subgraph  and the disjoint union of the rest of the two graphs. We denote the new circuit by $G_1 \star_{S_1} G_2$ when $S_2$ and the isomorphism is clear from context. See Figure~\ref{fig: glue} for an example. Note that (3) is a special case of (4). 

\end{definition}

%First
\begin{figure}
\begin{tikzpicture}[commutative diagrams/every diagram] 
\node (P2) at (2,1) {$\otimes$};
\node (P3) at (1,0) {\textbigcircle};;
\node (P4) at (3,-0) {\textbigcircle};;
\path[commutative diagrams/.cd, every arrow, every label] 
(P3) edge (P2)
(P4) edge (P2); 
%\draw [ ->](0.8,-0.8) to (1.8,1); 
%\draw [ ->] (3.2, -0.8) to (2.2,1);
%(P1) edge node[swap] {} (P0)  
%(P1) edge node {} (P0)  
%(P2) edge node[bend left] {} (P1)
%(P3) edge node[swap] {} (P2);
\node (P5) at (2,2) {$G_1$};;
\node (P5) at (7,3) {$G_2$};;
\node (P5) at (12,4) {$G_1 \boxplus G_2$};;

\node (P6) at (5,0) {\textbigcircle};;
\node (P7) at (7,0) {\textbigcircle};;
\node (P8) at (9,0) {\textbigcircle};;
\node (P9) at (6,1) {$\otimes$};
\node (P10) at (8,1) {$\oplus$};
\node (P11) at (7,2) {$\otimes$};
\path[commutative diagrams/.cd, every arrow, every label] 
(P6) edge (P9)
(P7) edge (P9)
(P7) edge (P10)
(P8) edge (P10)
(P9) edge (P11)
(P10) edge (P11);
%\draw [ ->](0.8,-0.8) to (1.8,1); 
%\draw [ ->] (3.2, -0.8) to (2.2,1)
\node (P12) at (11,0) {\textbigcircle
};;
\node (P13) at (10,1) {\textbigcircle};;
\node (P14) at (13,0) {\textbigcircle};;
\node (P15) at (12,1) {\textbigcircle};;
\node (P16) at (15,0) {\textbigcircle};;
\node (P17) at (13,1) {$\otimes$};
\node (P18) at (15,1) {$\oplus$};
\node (P19) at (13,2) {$\otimes$};
\node (P20) at (15,2) {$\otimes$};
\node (P21) at (14,3) {$\oplus$};
\path[commutative diagrams/.cd, every arrow, every label] 
(P12) edge (P17)
(P14) edge (P17)
(P14) edge (P18)
(P16) edge (P18)
(P13) edge (P19)
(P15) edge (P19)
(P17) edge (P20)
(P18) edge (P20)
(P19) edge (P21)
(P20) edge (P21); 
\end{tikzpicture}
\caption{Example of $G_1 \boxplus G_2$}
\label{fig: add}
\end{figure}
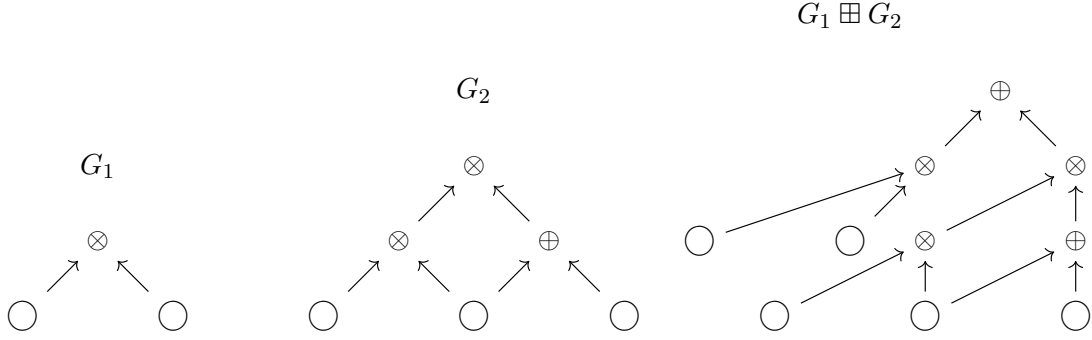

% Second example 
\begin{figure}
\begin{tikzpicture}[commutative diagrams/every diagram] 
\node (P1) at (1,0) {\textbigcircle};
\node (P2) at (2,0) {\textbigcircle};;
\node (P3) at (1,1) {$\oplus$};;
\node (P4) at (2,1) {$\otimes$};

\path[commutative diagrams/.cd, every arrow, every label] 
(P1) edge (P3)
(P2) edge (P3)
(P1) edge (P4)
(P2) edge (P4); 
%\draw [ ->](0.8,-0.8) to (1.8,1); 
%\draw [ ->] (3.2, -0.8) to (2.2,1);
%(P1) edge node[swap] {} (P0)  
%(P1) edge node {} (P0)  
%(P2) edge node[bend left] {} (P1)
%(P3) edge node[swap] {} (P2);
\node (P5) at (1.5,2) {$G_1$};;
\node (P5) at (5,2) {$G_2$};;
\node (P5) at (10,3) {$G_1 \curvearrowright G_2$};;
\node (P6) at (5,1) {$\otimes$};
\node (P7) at (4,0) {\textbigcircle};;
\node (P8) at (6,-0) {\textbigcircle};;
\path[commutative diagrams/.cd, every arrow, every label] 
(P7) edge (P6)
(P8) edge (P6); %\draw [ ->](0.8,-0.8) to (1.8,1); 
%\draw [ ->] (3.2, -0.8) to (2.2,1)
\node (P10) at (9,0) {\textbigcircle};
\node (P11) at (11,0) {\textbigcircle};;
\node (P12) at (9,1) {$\oplus$};;
\node (P13) at (11,1) {$\otimes$};
\node (P14) at (10,2) {$\otimes$};
\path[commutative diagrams/.cd, every arrow, every label] 
(P10) edge (P12)
(P11) edge (P12)
(P10) edge (P13)
(P11) edge (P13)
(P12) edge (P14)
(P13) edge (P14); 
\end{tikzpicture}
\caption{Example of $G_1\curvearrowright  G_2$}
\label{fig: cat}
\end{figure}
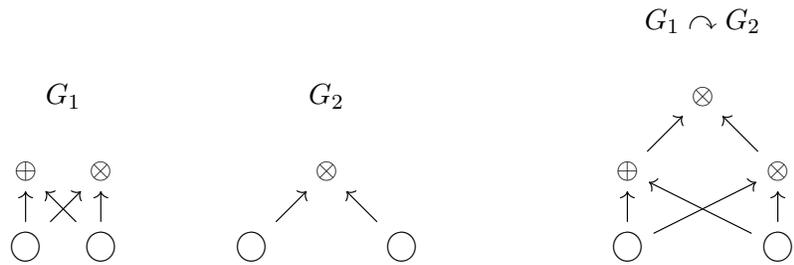

% Third example 
\begin{figure}
\begin{tikzpicture}[commutative diagrams/every diagram] 

%\draw [ ->](0.8,-0.8) to (1.8,1); 
%\draw [ ->] (3.2, -0.8) to (2.2,1);
%(P1) edge node[swap] {} (P0)  
%(P1) edge node {} (P0)  
%(P2) edge node[bend left] {} (P1)
%(P3) edge node[swap] {} (P2);
\node (P6) at (1,4) {$G$};;
\node (P6) at (10,4) {$G_S^{(2)}$};;

\node (P1) at (0,0) {\textbigcircle};
\node (P2) at (2,0) {\textbigcircle};;
\node (P3) at (0,1) {$\otimes$, $s_1$};;
\node (P4) at (2,1) {$\otimes$, $s_2$};
\node (P5) at (1,2) {$\oplus$};;
\node (P6) at (1,3) {$\otimes$};;
\path[commutative diagrams/.cd, every arrow, every label] 
(P1) edge (P3)
(P2) edge (P3)
(P1) edge (P4)
(P2) edge (P4)
(P3) edge (P5)
(P4) edge (P5)
%(P5) edge (P6)
(P5) edge (P6);
%\draw [ ->](0.8,-0.8) to (1.8,1); 
%\draw [ ->] (3.2, -0.8) to (2.2,1)
\node (P10) at (9,0) {\textbigcircle};
\node (P11) at (11,0) {\textbigcircle};;
\node (P12) at (9,1) {$\otimes$, $s_1$};;
\node (P13) at (11,1) {$\otimes$, $s_2$};
\node (P14) at (9,2) {$\oplus$};
\node (P15) at (11,2) {$\oplus$};
\node (P16) at (9,3) {$\otimes$};
\node (P17) at (11,3) {$\otimes$};
\path[commutative diagrams/.cd, every arrow, every label] 
(P10) edge (P12)
(P11) edge (P12)
(P10) edge (P13)
(P11) edge (P13)
(P12) edge (P14)
(P13) edge (P14)
(P12) edge (P15)
(P13) edge (P15)
(P14) edge (P16)
%(P14) edge[bend right] (P16)
%(P15) edge[bend left] (P17)
(P15) edge (P17);
\end{tikzpicture}
\caption{Example of $G_S^{(K)}$ for $K= 2$ and $S = \{s_1, s_2\}$}
\label{fig: repeat}
\end{figure}
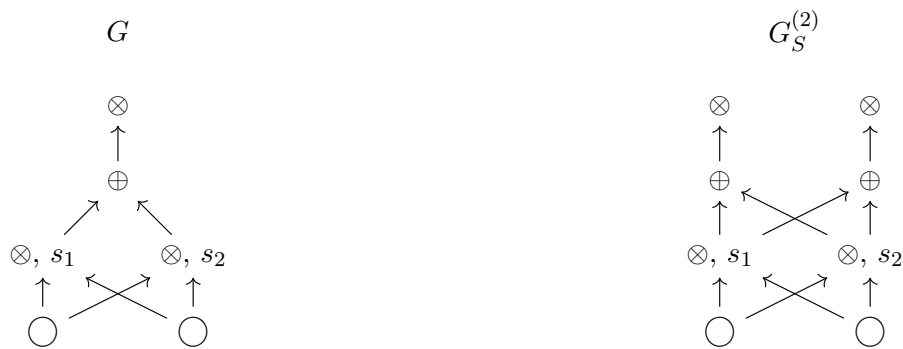

% Final example

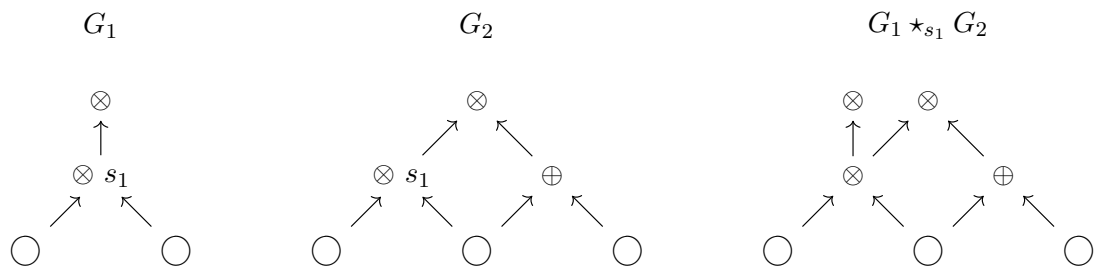
\begin{figure}
\begin{tikzpicture}[commutative diagrams/every diagram] 
\node (P1) at (2,2) {$\otimes$};
\node (P2) at (2,1) {$\otimes$ $s_1$};
\node (P3) at (1,0) {\textbigcircle};;
\node (P4) at (3,-0) {\textbigcircle};;
\path[commutative diagrams/.cd, every arrow, every label] 
(P3) edge (P2)
(P4) edge (P2)
(P2) edge (P1);
%(P2) edge[bend right] (P1); 
%\draw [ ->](0.8,-0.8) to (1.8,1); 
%\draw [ ->] (3.2, -0.8) to (2.2,1);
%(P1) edge node[swap] {} (P0)  
%(P1) edge node {} (P0)  
%(P2) edge node[bend left] {} (P1)
%(P3) edge node[swap] {} (P2);
\node (P5) at (2,3) {$G_1$};;
\node (P5) at (7,3) {$G_2$};;
\node (P5) at (13,3) {$G_1 \star_{s_1} G_2$};;
\node (P6) at (5,0) {\textbigcircle};;
\node (P7) at (7,0) {\textbigcircle};;
\node (P8) at (9,0) {\textbigcircle};;
\node (P9) at (6,1) {$\otimes$ $s_1$};
\node (P10) at (8,1) {$\oplus$};
\node (P11) at (7,2) {$\otimes$};
\path[commutative diagrams/.cd, every arrow, every label] 
(P6) edge (P9)
(P7) edge (P9)
(P7) edge (P10)
(P8) edge (P10)
(P9) edge (P11)
(P10) edge (P11);
%\draw [ ->](0.8,-0.8) to (1.8,1); 
%\draw [ ->] (3.2, -0.8) to (2.2,1)
\node (P12) at (11,0) {\textbigcircle};;
\node (P13) at (13,0) {\textbigcircle};;
\node (P14) at (15,0) {\textbigcircle};;
\node (P15) at (12,1) {$\otimes$ };
\node (P16) at (14,1) {$\oplus$};
\node (P17) at (13,2) {$\otimes$};
\node (P18) at (12,2) {$\otimes$};
\path[commutative diagrams/.cd, every arrow, every label] 
(P12) edge (P15)
(P13) edge (P15)
(P13) edge (P16)
(P14) edge (P16)
(P15) edge (P17)
(P16) edge (P17)
%(P15) edge[bend left] (P18)
(P15) edge (P18); 
\end{tikzpicture}
\caption{Example of $G_1 \star_{s_1} G_2$}
\label{fig: glue}
\end{figure}
% End of examples. 

Now we are ready to state our main theorem. Consider a knapsack problem with parameters $v_i (1 \leq i \leq n), w_i (1\leq i\leq n)$ and $W$. 

\iffalse
{\color{red} Hao: one SODA reviewer pointed out a mistake. I think I can fix it by letting 
$M := W + \sum w_i + \sum_i v_i$. Take $T = 2M\log M$, $k_r = 6 M \log M \log(M \log M)$, and $K = 4 \log(M \log M)$. 
In this case, $K$ is still logarithm in the $v_i$ and $W$'s. We can verify that $KT > k_r$, $k_r > T\log T + W' \log W'$. Moreover, we have $Kr_i - k_r \leq 4 \log(M \log M) M \log M - k_r \leq -2 M \log M \log(M\log M) < -M < v_i$, for all $i$. So in the proof we don't have to scale up $K$.}
\fi
\begin{theorem} \label{thm}
There exists a circuit $G = G(v_i, w_i, W)$, and integers $k_m, k_r$ such that 

(1) $G$ has  $O(polylog(v_i, w_i, W) ) \cdot poly(n))$ vertices. 

(2) $k_m, k_r = O( poly(v_i, w_i, W, n))$. 

(3) There exists a set of $n$ vertices $s_1, \ldots, s_n$ in $G$, such that if the length $l_{new}^*(i)$ is the length of $s_i$ in an optimal solution to the relinearize problem on $G$. Then $l_{new}^*(i) (1 \leq i \leq n)$ is an optimal solution to 
\[
\max \sum v_i l_i,  \mbox{ s.t. }  l_i \in \{1, 2\}, \sum w_i l_i \leq W + \sum w_i, 
\]
Hence $l_{new}^*(i) -1 (1 \leq i \leq n)$ is an optimal  solution to the original knapsack problem. 
\end{theorem}

Since our proof is long, we will break it into several parts. First, let $K, T$ be positive integers whose values will be determined later. We  define a circuit 
 \[
 	G^{0}:= \{ ((\cL(w_1) \boxtimes \cL(w_2) ) \cdots \boxtimes \cL(w_n) ) \boxplus \cL(W') ) \curvearrowright \cL(T) \}_S ^{(K)}
 \]
Here $W' = W + \sum_i w_i$, and $S = \{s_1, \ldots, s_n\}$, where $s_i$ is the first non-input vertex in the circuit $\cL(w_i)$. In particular, with no relinearization the length of $s_i$ is equal to 2. Consider the relinearize problem on the circuit $G^0$ and let $l_i$ be the new lengths of $s_i$. Without loss of generality,  we assume that $w_i \leq W$ for all $i$ (if $w_i> W$,  then  any optimal solution of the knapsack problem always have $x_i = 0$, and we can reduce the dimension of the problem by one). 

\begin{lemma} \label{lem: constraint}
Suppose $$k_r >  4k_m(T\log T +W' \log W'),$$  and $k_m =1$. Then for any optimal solution to the relinearize problem on $G^0$, the only vertices that could have nonzero relinearization are the $s_i$. 
\end{lemma}

\begin{proof}
By Lemma~\ref{lem: one},  the total cost of evaluating a circuit of type $\cL(T)$ is bounded by $4 k_m T\log T$, hence relinearizing any single vertex in this circuit has benefit bounded by $4k_m T\log T$. The situation is similar for $\cL(W')$. Note that relinearizing verteices in $\cL(W')$ 
could reduce the length of vertices in $\cL(T)$, but the benefit is still bounded above by $4k_m(T\log T + W'\log W')$. For the same reason, the benefit of relinearizing any vertex in any of the $K$ copies of $\cL(w_i)$ is bounded by $4k_m(T\log T + w_i \log w_i)$. Since $w_i \leq W'$, this completes
the proof.
\end{proof}

\begin{lemma} \label{lem: constraint}
Suppose $k_m KT > k_r$.  Then for any optimal solution to the relinearization problem on $G^0$ we must have 
\[
	\sum l_i w_i \leq  W':= W + \sum_{i =1}^n w_i. 
\]
Here again we recall that $l_i \in \{1, 2 \}$ denote the length of $s_i$ in an optimal solution. 
\end{lemma}

\begin{proof}
Suppose the claim is false. Then there exists $i$ such that $l_i = 2$. We relinearize the vertex $s_i$, which reduces the length of the final output in each  copy of $\cL(w_i)$ by $w_i$, and the length of the output vertex of
$$
(\cL(w_1) \boxtimes \cL(w_2)) \cdots \boxtimes \cL(w_n))$$

is reduced by $w_i$. Since $\sum l_i w_i > W'$, the length of the input vertex in each $\cL(T)$ is reduced by at least one, and the cost reduction from each $\cL(T)$ is at least $k_m T$. Hence we the benefit 
we collect from relinearizing $s_i$ is at least $k_m KT$, whereas the cost is $k_r$. Since we assumed $k_m KT > k_r$, we know relinearizing the vertex $s_i$ reduces  the total cost. This is a contradiction, since we started with an optimal solution. 
\end{proof}

Now we can starting proving Theorem~\ref{thm}.

\begin{proof} (of Theorem~\ref{thm}) 
Let 
$M = W + \sum_i w_i + \sum_i v_i$. We take $T = \lceil 5M \log M \rceil$, $k_r = 25 \lceil M \log M \log(M \log M) \rceil$,  $K = 6 \lceil \log(M \log M) \rceil$ and $k_m = 1$. It is easy to see that $K,T, k_r$ are of size polynomial in $W, w_i, v_i$. One can verify that 
$k_m KT > k_r$ and $k_r > 4k_m( T\log T + W' \log W')$. Thus, by Lemma~\ref{lem: constraint}, we have $\sum l_i w_i \leq W'$ if $l_i$ are the new length of $s_i$ in any optimal solution to the relinearization problem on $G^0$. This means we have the correct constraint. However, the costs are wrong: the total cost of evaluating the circuit $G^0$ is given by 
\[
	K(\sum r_i l_i) +\sum_i k_r (2 - l_i) + C, 
\]
where as we proved in Lemma~\ref{lem: one}, $r_i  \leq 4 w_i \log(w_i)$. Here $C$ is the cost of evaluating all the $\cL(T)$ circuits plus all the $\cL(W')$ circuits. The fact that $C$ is a constant follows from Lemma~\ref{lem: constraint}.

Note that the coefficient before $l_i$ is equal to $Kr_i - k_r$, and we want to modify this coefficient to $-v_i$. First, note that 

\begin{align*}
Kr_i - k_r &\leq K 4 w_i \log w_i - k_r  \\
&\leq 4 K \lceil M \log M \rceil - k_r  \\
&\leq (24 - 25) \lceil M \log M \log(M\log M) \rceil \\
&\leq -M  \\
&\leq -v_i, \forall i. 
\end{align*}
Let $\lambda_i = k_r - Kr_i - v_i \in \bZ_{\geq 0}$. We claim that there exists a circuit $\cL'(\lambda_i)$ of  such that relinearizing its first non-input vertex  reduces the total multiplication cost by $\lambda_i$. We omit the details of construction of $\cL'$ since it is similar to that of $\cL$. In particular, 
the $\cL'(\lambda_i)$ can be constructed with at most $2 \log (\lambda_i)$ vertices. We then let 
\[
	G^1 = G^0 \star_{s_1} \cL'(\lambda_1), \ldots, 	G^i = G^{i-1} \star_{s_i} \cL'(\lambda_i), \ldots, G^n = G^{n-1} \star_{s_n}  \cL'(\lambda_n)
\]
and set $G = G^n$. Since $\lambda_i < k_r$, one can see that in any optimal solution of the relinearize problem on $G$, the vertices in $\cL'(\lambda_i)$ have zero relinearization. Thus, the relinearize problem on $G$ is equivalent to 
\[
\min \sum_{i=1}^{n} -v_i l_i + C',  \mbox{ s.t. }  l_i \in \{1, 2\} \mbox{ and } \sum_{i=1}^{n} w_i l_i \leq W + \sum w_i, 
\]
which is equivalent to 
\[
\max \sum_{i=1}^{n} v_i l_i,  \mbox{ s.t. }  l_i \in \{1, 2\} \mbox{ and } \sum_{i=1}^{n}  w_i l_i \leq W + \sum w_i.
\]
This proves part (3) of Theorem~\ref{thm}. Part (1) is clear since the number of vertices in $G$ is bounded by $2K ( \log(T) + \log(W') +  \sum_{i=1}^n \log(w_i) ) + 2\sum_{i=1}^n \log(\lambda_i)$. Hence it is logarithm in the parameters $v_i, w_i, W$ and linear in the number of variables $n$. For (2), note that we set $k_m =1$, so it suffices to prove it for $k_r$. By construction,  $k_r$ is also bounded by a polynomial in $v_i, w_i, W$. This completes the proof.
\end{proof}

\begin{corollary}
The relinearize problem is NP-hard. 
\end{corollary}

\section{An Simple Case}
Assume we are in the situation where each non-input vertex in the circuit has two inputs and at most one output. 
In this case, we have a polynomial time algorithm for the relinearize problem. For a vertex $i$, define $M(i,\ell)$ to be the minimal cost to compute the circuit up to vertex $i$, so that the new length of $i$ is $\ell$.

% Moreover, we have a linear time algorithm that achieves an approximation guarantee of 

Recall that $p_1(i)$ and $p_2(i)$ denote the parents of $i$. If $i$ is a multiplicative vertex, we have 
$$M(i,\ell)=\min_{\ell_1,\ell_2}\{M(p_1(i),\ell_1)+M(p_2(i),\ell_2)+k_r(\ell_1+\ell_2-\ell)+k_m(\ell_1+\ell_2)\}.$$

If $i$ is an addition vertex, we have 
$$M(i,\ell)=\min_{\ell_1,\ell_2}\{M(p_1(i),\ell_1)+M(p_2(i),\ell_2)+k_r(\max\{\ell_1,\ell_2\}-\ell)\}.$$

Here it is important that the vertices all only have a single output, since otherwise $p_1(i)$ and $p_2(i)$ might have a common ancestor, in which case relinearizing this ancestor might benefit both of them.

\begin{claim} Suppose $N  = |V| \ge 2$. Then in the above formulae, it suffices to take the minimum over range $2\le \ell_1, \ell_2 \le N$.
\end{claim}

\begin{proof}
For the input vertices, the lengths is at most $2$. For any non-input vertex $v$, we prove inductively that its length cannot exceed its number of  ancestors. The length is at most $l(p_1(v)) + l(p_2(v))-1$, and by inductive hypothesis, both $l(p_1(v))$ and $l(p_1(v))$ are at most their number of ancestors (or plus one if it happens to be an input vertex). That is, $l(p_1(v)) + l(p_2(v))-1 \le n_1+n_2+1=n$. Here $n_1, n_2, n$ denote the number of ancestors for $p_1(v), p_2(v), v$, respectively. 
\end{proof}

Now our algorithm proceeds as follows. We traverse the $N$ vertices. At each vertex, we compute $M(i,l)$ for $O(N)$ values of $l$, and each computation requires $O(N^2)$ operations. Thus the total running time is $O(N^4)$. Finally, the optimal cost is given by 
$\min_{2 \leq l \leq N} M(v, l)$, where $v$ is the output node of the graph $G$. 

\section{Conclusion and Future Work}

Fully homomorphic encryption evaluates boolean circuits, and relinearization is a standard technique to reduce the ciphertext sizes after evaluation. In this paper, we consider the goal of optimizing where and how much to perform the relinearization operation in any given circuit, in order to minimize the total computational cost. We formalized it as a discrete optimization problem, and proved that the problem is NP-hard. In the special case where every node has at most one ouptut node, we give a polynomial time algorithm. 

For future directions, it is of interest to design fast approximate algorithms for the relinearization problem. Also, one can aim at optimizing specific circuits that appear in the literature for applications of FHE. Examples include  components of the AES encryption/decryption circuit and machine learning models such as logistic regression or neural network. 

\nocite{*}
\bibliographystyle{alpha}
\bibliography{relinearization}

\end{document}